\pdfoutput=1 

\documentclass{article}
\usepackage[utf8]{inputenc}
\usepackage[english]{babel}
\usepackage{amsmath,amsfonts}
\usepackage{amssymb}
\usepackage[showonlyrefs]{mathtools}
\usepackage{amsthm}
\usepackage{moresize}
\usepackage[a4paper, portrait, left=1in, right=1in, top=1in, bottom=1in]{geometry}
\usepackage[numbers,sectionbib]{natbib}
\usepackage[hyperindex,breaklinks,bookmarks,colorlinks,urlcolor=blue,citecolor=blue,linkcolor=blue]{hyperref}
\hypersetup{bookmarksdepth=2}
\usepackage[shortlabels]{enumitem}
\usepackage[toc]{appendix}
\usepackage{cleveref}
\usepackage[ruled,vlined, linesnumbered]{algorithm2e}
\usepackage[dvipsnames]{xcolor}
\newcommand{\eps}{\varepsilon}
\usepackage{todonotes}

\DeclareMathOperator*{\E}{E}
\DeclareMathOperator*{\Var}{Var}
\DeclareMathOperator*{\Varht}{Var_{\mathrm{HH}}}

\newcommand{\calP}{\mathcal{P}} 
\newcommand{\calQ}{\mathcal{Q}} 
\renewcommand{\epsilon}{\eps}   
\newcommand{\usize}{N}          
\newcommand{\esize}{n}          
\newcommand{\ssize}{m}          
\newcommand{\estimator}{\zeta}  
\newcommand{\htest}{\mu_{\textrm{HH}}} 

\renewcommand{\th}{{}^{\mathrm{th}}}

\newcommand{\set}[1]{\left\{#1\right\}}
\newcommand{\sucht}{\,\colon\,}
\newcommand{\paren}[1]{\left(#1\right)}
\newcommand{\abs}[1]{\left|#1\right|}

\newtheorem{theorem}{Theorem}
\newtheorem{corollary}[theorem]{Corollary}
\newtheorem{lemma}[theorem]{Lemma}
\theoremstyle{definition}
\newtheorem{definition}[theorem]{Definition}
\newtheorem{example}[theorem]{Example}

\theoremstyle{remark}

\date{}
\title{
Bias Reduction for Sum Estimation
}

\author{
    Talya Eden\\
	\texttt{talyaa01@gmail.com}\\
	Bar-Ilan University
	\and
	Jakob Bæk Tejs Houen\\
	\texttt{jakob@tejs.dk}\\
	BARC, Univ. of Copenhagen
	
	\and
	Shyam Narayanan\\
	\texttt{shyamsn@mit.edu}\\
	MIT
	    
	\and
        Will Rosenbaum\\
	\texttt{wrosenbaum@amherst.edu}\\
	Amherst College
	\and
	Jakub Tětek\\
	\texttt{j.tetek@gmail.com}\\
	BARC, Univ. of Copenhagen
}

\begin{document}
\maketitle

\begin{abstract}
  In classical statistics and distribution testing, it is often assumed that elements can be sampled exactly from some distribution $\calP$, and that when an element $x$ is sampled, the probability $\calP(x)$ of sampling $x$ is also known. In this setting, recent work in distribution testing has shown that many algorithms are robust in the sense that they still produce correct output if the elements are drawn from any distribution $\calQ$ that is sufficiently close to $\calP$. This phenomenon raises interesting questions: under what conditions is a ``noisy'' distribution $\calQ$ sufficient, and what is the algorithmic cost of coping with this noise?

  In this paper, we investigate these questions for the problem of estimating the sum of a multiset of $\usize$ real values $x_1, \ldots, x_\usize$. This problem is well-studied in the statistical literature in the case $\calP = \calQ$, where the Hansen-Hurwitz estimator [Annals of Mathematical Statistics, 1943] is frequently used. We assume that for some (known) distribution $\calP$, values are sampled from a distribution $\calQ$ that is pointwise close to $\calP$. That is, there is a parameter $\gamma < 1$ such that for all $x_i$, $(1 - \gamma) \calP(i) \leq \calQ(i) \leq (1 + \gamma) \calP(i)$. 
  For every positive integer $k$ we define an estimator $\zeta_k$ for $\mu = \sum_i x_i$ whose bias is proportional to $\gamma^k$ (where our $\zeta_1$ reduces to the classical Hansen-Hurwitz estimator).
  As a special case, we show that if $\calQ$ is pointwise $\gamma$-close to uniform
  and all $x_i \in \{0, 1\}$, for any $\eps > 0$,  we can estimate $\mu$ to within additive error $\eps \usize$ using $m = \Theta({\usize^{1-\frac{1}{k}} / \eps^{2/k}})$ samples, where $k = \left\lceil (\lg \eps)/(\lg \gamma)\right\rceil$. We then show that this sample complexity is essentially optimal. Interestingly, our upper and lower bounds show that the sample complexity need not vary uniformly with the desired error parameter $\eps$: for some values of $\eps$,  perturbations
  in its value have no asymptotic effect on the sample complexity, while for other values, any decrease in its value results in an asymptotically larger sample complexity.

\end{abstract}


 \thispagestyle{empty}
\newpage
\clearpage
\setcounter{page}{1}

\section{Introduction}
\label{sec:introduction}
Consider the following simple problem. Let us have values $x_i \in \set{0, 1}$ for $i \in [\usize]$ and assume we may sample $i$ from a distribution $\calQ$ that is pointwise $\gamma$-close to uniform (see Definition~\ref{dfn:pointwise}). 
It is easy to obtain an additive $\pm \gamma N$ approximation of the number of $1's$. But is it possible to get a better approximation using a number of samples that is sub-linear in $\usize$? We answer this question positively. Specifically, we solve a more general sum estimation problem, with the above problem being the simplest application. Additionally, we derive lower bounds showing that for a wide range of parameters, the sample complexity of our algorithm is asymptotically tight.

In full generality, estimating the sum of a (multi)set of numbers is a fundamental problem in statistics, and the problem plays an important role in the design of efficient algorithms for large datasets.
The basic problem can be formulated as follows: given a multiset of $\usize$ elements, $S = \set{x_1, x_2, \ldots, x_\usize}$, compute an estimate of the sum $\mu = \sum_{i \in [\usize]} x_i$.


Assume that the values $x_i$ can be sampled according to some probability distribution $\calQ$ over $S$ (equivalently, over $[\usize]$). The classical work of Hansen and Hurwitz~\cite{hansen1943theory} examines the setting in which, when an element $x \in S$ is sampled, the probability $\calQ(x)$ can be determined. They introduce the Hansen-Hurwitz estimator defined by
\begin{equation}
  \label{eqn:hansen-hurwitz}
  \htest = \frac{1}{\ssize} \sum_{j = 1}^{\ssize} \frac{X_j}{\calQ(X_j)}
\end{equation}
where $X_1, X_2, \ldots, X_\ssize$ are samples taken from the distribution $\calQ$. This estimator has been used extensively (though often implicitly) in sublinear algorithms \cite{cohen2014algorithms,Tetek2022-edge,berenbrink2014estimating,holm2022massively}. Hansen and Hurwitz prove that~(\ref{eqn:hansen-hurwitz}) is an accurate estimator of $\mu$ via the following theorem:

\begin{theorem}[Hansen \& Hurwitz, 1943~\cite{hansen1943theory}]
  \label{thm:hansen-hurwitz}
  The value $\htest$ is an unbiased estimator of the sum $\mu$ (i.e., $\E(\htest) = \mu$) and its variance is
  \begin{equation}
    \label{eqn:ht-var}
      \Var[\htest] = \frac{1}{\ssize}\sum_{i = 1}^\usize \calQ(i) \paren{\frac{x_i}{\calQ(i)} - \mu}^2.
  \end{equation}
\end{theorem}

Theorem~\ref{thm:hansen-hurwitz} can be applied, to obtain probabilistic guarantees for estimating $\mu$ via sampling. For example, if one wishes to compute a $1 \pm \eps$ multiplicative estimate of $\mu$ with probability $1 - \delta$, by the Chebyshev inequality, it suffices to take $\ssize$ sufficiently large that $\Var[\htest] / (\eps^2 \mu^2) < \delta$. 

In practice, it may however be unreasonable to assume that the probability distribution from which elements are sampled is known precisely. For example, the underlying process generating the samples may be noisy or may induce some underlying bias. We model this situation by assuming that the true sampling distribution $\calQ$ is close to some known distribution $\calP$. When an element $x$ is sampled, the probability $\calP(x)$ can be determined, but not the true probability $\calQ(x)$. We assume that $\calQ$ is \emph{pointwise close} to $\calP$ in the following sense.

\begin{definition}
  \label{dfn:pointwise}
  Let $\calP$ and $\calQ$ be probability distributions over a (multi)set $S$. Then for any $\gamma < 1$, we say that $\calQ$ is \emph{pointwise $\gamma$-close} to $\calP$ if for every $x \in S$, we have
  \begin{equation}
    \label{eqn:pointwise}
    (1 - \gamma) \calP(x) \leq \calQ(x) \leq (1 + \gamma) \calP(x).
  \end{equation}
\end{definition}

Given the situation above, one can apply the Hansen-Hurwitz estimator~(\ref{eqn:hansen-hurwitz}) with the known distribution $\calP$ in place of the true sample distribution $\calQ$. We define the \emph{positive sum}, $\mu_+$ to be
\begin{equation}
  \label{eqn:pos-sum}
  \mu_+ = \sum_{x \in S} \abs{x}.
\end{equation}
It is straightforward to show that the resulting estimator has bias at most $\gamma\mu_+$, and its variance increases by a factor of $1 + O(\gamma)$. 
However, the parameter $\gamma$ may be too large to guarantee the desired error in the estimate of $\mu$. 
For the above problem of estimating the sum of $0$-$1$ values, this would lead to error of $\gamma N$, while we want error of $\epsilon N$ for $\epsilon < \gamma$.\footnote{It is possible to get tighter bounds if parameterizing also by the sum, but for simplicity, we choose to parameterize the error only by $N,\epsilon,$ and $\gamma$.}

Our setting is closely related to recent work in distribution testing. For example, it has been noted that many algorithms that rely on a probability oracle are ``robust'' in the sense that we may do distribution testing to within $\epsilon$ if the oracle's answers have relative error of, say, $1\pm\epsilon/3$ \cite{onak2018-probability, canonne2014testing}.
Our work goes further in the sense that our estimators work also in the setting when the error in the oracle's answers is greater than the desired error parameter $\epsilon$. Specifically, our goal is to characterize the (sample) complexity of a task as a function the oracle error parameter $\gamma$ and a desired approximation parameter $\eps$. This can also be seen as a generalization of the learning-augmented distribution testing setting where $\gamma$ is assumed to be constant~\cite{eden2021-learning}.

\subsection{Our Contributions}

In this paper, our goal is to estimate the sum $\mu = \sum_{i=1}^{N} x_i$ with an error that is strictly less than the bias $\gamma \mu_+$ (Equation~(\ref{eqn:pos-sum})) guaranteed by $\htest$. Specifically, given a desired error parameter $\eps$ with $0<\eps<\gamma$, we wish to estimate $\mu$ with bias close to $\eps \mu_+$.
In our setting, for each sample we are given a random index $i \in [N]$ drawn from the unknown distribution $\calQ$, along with the value $x_i$ and our estimate $\calP(i)$ of the true probability $\calQ(i)$.
We introduce a family of estimators $\estimator_1, \estimator_2, \ldots$, where each $\estimator_k$ has bias at most $\gamma^k \mu_+$. To motivate the construction of $\estimator_k$, we first re-write the Hansen-Hurwitz estimator in terms of the frequency vector of samples from $S$. Specifically, if $X_1, X_2, \ldots, X_\ssize$ are the sampled elements, define the frequency vector $Y = (Y_1, Y_2, \ldots, Y_\usize)$ by
\[
  Y_i = \abs{\set{j \sucht X_j = i}}.  
\]
We define the estimator
\[
\xi_1 = \frac{1}{\ssize} \sum_{i = 1}^\usize Y_i \cdot \frac{x_i}{\calP(i)}.
\]
Note that this estimator can be efficiently implemented, as the items that have not been sampled contribute $0$ to the sum. We may thus implement this in time linear in the sample complexity, and do not need to take $O(\usize)$ time.

In the case where $\calP = \calQ$, $\xi_1$ is equivalent to the Hansen-Hurwitz estimator $\htest$. More generally, $\calQ$ is pointwise $\gamma$-close to $\calP$, and $\xi_1$ has bias at most $\gamma\mu_+$.

The estimator $\xi_1$ can be generalized as follows. Rather than sampling individual elements, we can examine $h$-wise collisions between samples, where an \emph{$h$-wise collision} consists of $h$ samples resulting in the same outcome.

\begin{definition}
  \label{eqn:h-wise-est}
  For any positive integer $h$, we define the \emph{$h$-wise collision estimator} $\xi_h$ of $\mu = \sum_i x_i$ to be
  \begin{equation}
    \label{eqn:xi-estimator}
    \xi_h = \frac{1}{\binom \ssize h} \sum_{i = 1}^{\usize} \binom{Y_i}{h} \frac{x_i}{(\calP(i))^h}.
  \end{equation}
\end{definition}

We note that $\binom{Y_i}{h}$ gives the number of $h$-wise collisions involving the value $X_j = i$. It is straightforward to show that when $\calQ = \calP$, all $\xi_h$ are unbiased estimators for $\mu$, and that $\xi_h$ has bias $O(h \gamma \mu_+)$ when $\calQ$ is pointwise $\gamma$-close to $\calP$.

Individually, the estimators $\xi_1, \xi_2, \ldots$ are no better than $\xi_1 = \htest$ in terms of bias and variance. As we will show, however, for any positive integer $k$, a suitable linear combination of the $\xi_i$ can be chosen such that the coefficients of $\gamma^j$ in the bias cancel out for $j < k$. The resulting estimator then has bias $\leq \gamma^k \mu_+$.

\begin{definition}
  \label{dfn:estimator}
  For each positive integer $k$, we define the \emph{bias reducing estimator of order $k$}, $\estimator_k$, to be
  \begin{equation}
    \label{eqn:estimator}
    \estimator_k = \sum_{h = 1}^k (-1)^{h+1} \binom k h \xi_h = \sum_{h = 1}^k (-1)^{h+1} \frac{\binom k h}{\binom m h} \sum_{i \in [\usize]} \binom{Y_i}{h} \frac{x_i}{(\calP(i))^h}.
  \end{equation}  
\end{definition}

\begin{example}
In order to give some intuition about the expression~(\ref{eqn:estimator}), consider the case where $\mathcal{P}$ is the uniform distribution and $k = 2$. Define $\alpha_i$ to be such that $\calQ(i) = (1+\alpha_i) \calP(i)$. Note that $|\alpha_i| \leq \gamma$ because $\calQ$ is assumed to be pointwise $\gamma$-close to uniform. By a simple calculation, we have that $\E[\xi_1] = \sum_{i=1}^\usize (1+\alpha_i)x_i$ and $\E[\xi_2] = \sum_{i=1}^\usize (1+\alpha_i)^2x_i = (1+2\alpha_i+\alpha_i^2) x_i$. Therefore, it holds that 
\[
\E[\estimator_2] = \E[2 \xi_1 - \xi_2] = \sum_{i=1}^\usize  2(1+\alpha_i)x_i - (1+2\alpha_i+\alpha_i^2) x_i = \sum_{i=1}^n (1+\alpha_i^2) x_i \subseteq \mu \pm \sum_{i=1}^n\alpha_i^2 |x_i| \subseteq \mu \pm \gamma^2 \mu_+. \footnote{The first of the two inclusions is not tight in that the absolute value in $|x_i|$ is not necessary, but it becomes necessary for odd values of $k$.}
\]
This proves that the bias of the estimator is at most $\gamma^2 \mu_+$. Similarly, one can show that the bias of the estimator above is $\leq \gamma^k \mu_+$. We prove this in \Cref{sec:sum-estimation}.
\end{example}


\begin{theorem}[Bias portion of \Cref{thm:estimate-sum}]
  \label{thm:estimator-bias-reduction}
  Supoose $\calP$ and $\calQ$ are probability distributions over $[\usize]$ with $\calQ$ pointwise $\gamma$-close to $\calP$. Let $\estimator_k$ be defined as in~(\ref{eqn:estimator}). Then
  \[
  \abs{\E[\estimator_k - \mu]} \leq \gamma^k \mu_+.
  \]
  In particular, if $x_i \geq 0$ for all $i$, then $\estimator_k$ has bias at most $\gamma^k \mu$.
\end{theorem}

This theorem shows that $\estimator_k$ reduces the bias to $\gamma^k$ compared to the bias $\gamma$ for the Hansen-Hurwitz estimator (equivalent to $\estimator_1$). \Cref{thm:estimate-sum} additionally bounds the variance of $\estimator_k$, which is required for our applications. 

We apply Theorem~\ref{thm:estimator-bias-reduction} (or more specifically, \Cref{thm:estimate-sum}) to obtain our main algorithmic results. Our goal is as follows: given sample access to some $\calQ$ that is pointwise $\gamma$-close to $\calP$ and a desired error parameter $\eps$, estimate $\mu$ with error $\eps$ using as few samples as possible. To this end, we employ a two-stage estimatotion technique (see Algorithm~\ref{alg:improved-estimate}). In the first stage, we use the 1-wise collision estimator $\xi_1$ (i.e., the Hansen-Horwitz estimator) to obtain a coarse estimate of $\mu$. Then, the second stage refines this estimate by applying the bias reducing estimator $\estimator_k$ with an appropriately chosen $k$. Specifically, we show the following:

\begin{theorem}[Special case of \Cref{thm:estimate-sum}]
  \label{thm:intro-estimate-sum}
  Define $\esize = \max_i 1 / \calP(i)$.\footnote{Note that $\usize \leq \esize$, where $\usize$ is the size of the multiset being sampled. In the case where $\calP(i) = \Omega(1/\usize)$ for all $i$, we have $\esize = \Theta(\usize)$. The convention of defining $\esize$ in this way was previously used in~\cite{eden2021-learning}.} Suppose $\calQ$ is pointwise $\gamma$-close to $\calP$, and let $\Varht$ denote the variance of the Hansen-Hurwitz estimator (Equation~\ref{eqn:ht-var}). For $\eps_1,\eps_2 > 0$, define $k = \lceil (\log \eps_1)/\log \gamma)\rceil$. Then using
  \[
  m = O\left(\sqrt[k]{n^{k-1} \eps_2^{-2} \Varht} \right)
  \]
   independent samples from $\calQ$, with probability at least $2/3$, Algorithm~\ref{alg:improved-estimate} produces an estimate $\hat\mu$ of $\mu = \sum_i x_i$ 
   with absolute error
   \[
  \abs{\mu - \hat\mu} \leq \eps_1 \mu_+ + \eps_2
  \]
\end{theorem}

To understand the complexity of this algorithm, we note that when $\mathcal{P} = \mathcal{Q}$, in order to get an error $\epsilon_2$ the complexity of the Hansen-Hurwitz estimator is $\eps_{2}^{-2} \Varht$. The complexity of our algorithm can thus be seen as a weighted geometric average between the complexity of the Hansen-Hurwitz estimator, and $\esize$.

As a corollary of \Cref{thm:intro-estimate-sum}, we obtain a solution to the aforementioned problem of estimating a sum of $0$-$1$ values.

\begin{corollary}
  \label{cor:zero-one-sum}
  Suppose $\calQ$ is pointwise $\gamma$-close to the uniform distribution over $[\usize]$ and $x_i \in \set{0, 1}$ for every $i \in [\usize]$. For any $\eps > 0$ define $k = \lceil (\log \eps) / \log \gamma \rceil$. Then $m = O(n^{1 - 1/k} \eps^{-2/k})$ samples are sufficient to obtain an estimate of $\mu = \sum_{i} x_i$ with additive error $\eps \usize$ with probability $2/3$. 
\end{corollary}

We note that the asymptotic sample complexities in Theorem~\ref{thm:intro-estimate-sum} and Corollary~\ref{cor:zero-one-sum} are non-uniform in $\gamma$ and $\eps$. In the case of Corollary~\ref{cor:zero-one-sum}, for any fixed postive integer $k$ and constant $\gamma > 0$, if $\eps = \gamma^k$, then $O(n^{1 - 1/k})$ samples are sufficient to obtain an $\eps \usize$ additive estimate. On the other hand, if $\gamma^{k+1} \leq \eps < \gamma^k$, then our algorithm uses $O(n^{1 - 1/(k+1)})$ samples. Our next main result shows that this sample complexity is essentially optimal, and perhaps surprisingly, that the non-uniformity of the sample complexity is unavoidable. Specifically, we show the following lower bound.

\begin{theorem}
  \label{thm:lb-intro}
  Let $\usize$ be a positive integer and $x_1, x_2, \ldots, x_\usize \in \set{0, 1}$. Suppose $A$ is any algorithm such that for any distribution $\calQ$ that is pointwise $\gamma$-close to uniform, and any $\eps > 0$, $A$ produces an estimate $\hat\mu$ of $\mu$ with additive error at most $\eps \usize$ with probability $2/3$ using samples from $\calQ$. Then for every positive integer $k$, there exists an absolute constant $c_k$ with $0 < c_k < 1$ such that for $\eps \leq c_k \gamma^k$, $A$ requires $\Omega(n^{1 - 1/(k+1)})$ samples from $\calQ$. 
\end{theorem}

This lower bound matches the upper bound of Corollary~\ref{cor:zero-one-sum} for a large range of parameters. When $\gamma \leq c_k$ (where $c_k$ is as in the conclusion of the theorem), our algorithm has sample complexity $O(n^{1 - 1/(k+1)})$ for all $\eps \in [\gamma^{k+1}, \gamma^k)$, while the lower bound shows $\Omega(n^{1 - 1/(k+1)})$ samples are necessary for all $\eps \in [\gamma^{k+1}, c_k \gamma^k]$. Interestingly, these matching upper and lower bounds show that the asymptotic sample complexity is non-uniform as a function of $\epsilon$ for any fixed (sufficiently small) $\gamma$: for every $\eps \in [\gamma^{k+1}, c_k \gamma^k]$, exactly $\Theta(n^{1 - 1/(k+1)})$ samples are necessary and sufficient, while for $\eps = \gamma^k$, $\Theta(n^{1 - 1/k})$ samples are necessary and sufficient. Thus, as a function $\eps$, the sample complexity contains ``islands of stability''---intervals in which some perturbations of $\eps$ have
no effect on the asymptotic sample complexity---while between these intervals, an arbitrarily small (constant) decrease in $\eps$ results in a polynomial (in $\usize$) increase in the sample complexity. 

\paragraph{Discussion and Related Work.}

Throughout the paper, we assume that the probability distribution $\calQ$ from which samples are generated is pointwise close to $\calP$ in the sense of Definition~\ref{dfn:pointwise}. Pointwise closeness is a strictly stronger (and less commonly used) distance measure than, for example, total variation distance (i.e., $L_1$ distance) or other $L_p$ distances. Nonetheless, this relatively strong assumption about the relationship between $\calP$ and $\calQ$ is necessary in order to obtain any non-trivial guarantee for estimating $\mu$.\footnote{As an extreme example, consider the case where only a $\gamma$-fraction of values $x_i$ are nonzero. If $\calP$ is uniform, then $\calQ$ can assign zero mass to the nonzero elements, and still be $\gamma$-close to $\calP$ with respect to total variation distance. Thus, \emph{any} estimator will return a value that is independent of the actual sum.} Algorithms for generating samples with pointwise approximation guarantees have been studied in the context of sublinear time algorithms~\cite{Eden2018-sampling, onak2018-probability, Eden2019-arboricity, Eden2021-sampling, eden2021-learning, Tetek2022-edge, Eden2022-almost} as well as Markov chains~\cite{Morris2003-evolving, hermon2018sensitivity}. In the latter case, \emph{uniform mixing time} gives a bound on complexity of obtaining samples with pointwise guarantees. Interestingly, Hermon~\cite{hermon2018sensitivity} shows a result that can be viewed as an analogue of our lower bound for Markov chains (specifically random walks on bounded degree graphs). Namely, small perturbations in the transition probabilities of edges can result in an asymptotic increases in the uniform mixing time.

The problem of estimating the sum is well-studied in statistics. Classical estimators for non-uniform sampling probabilities are described by Hansen and Hurwitz~\cite{hansen1943theory} and Horvitz and Thompson~\cite{horvitz1952-generalization}.
Sum estimation in the related setting where we do not know the sampling probabilities but know that they are proportional to the items' values has been studied in~\cite{motwani2007estimating, beretta2022better}.

\paragraph{Open problem: Sample correctors for uniform sampling.} Finally, we state one interesting open problem. The concept of sample correctors from~\cite{Canonne2018-sampling} assumes that we may sample from a distribution that is close to some property, and we want to be able to use it to sample from a distribution even closer to satisfying the property. It is natural to ask if one can use $o(n)$ samples from a distribution pointwise $\gamma$-close to uniform and simulate a sample with bias $o(\gamma)$.


\subsection{Technical Overview}
\label{sec:technical-overview}

\paragraph{Upper Bound.}
The goal is to reduce the bias from $\gamma$ to $\gamma^k$.
Now the main observation that guides our construction is the following identity:
\begin{align}\label{eq:bias-identity}
  1 + (-1)^{k + 1} \gamma^k = \sum_{h = 1}^k (-1)^{h + 1} \binom{k}{h} (1 + \gamma)^h
    \; .
\end{align}
This should be compared to our estimator, $\estimator_k = \sum_{h = 1}^k (-1)^{h+1} \binom k h \xi_h$, which clearly mirrors the identity.
The reason for this is that the probability of an $h$-wise collision on position $i \in [\usize]$ is $\calQ(i)^h \approx (1 + \gamma)^h \calP(i)^h$ in the worst case when $\frac{\calQ(i)}{\calP(i)} = 1+\gamma$.
Hence the expectation of the $h$-wise collision estimator, $\xi_h$, is approximately bounded by $(1 + \gamma)^h \mu$.
This implies that when we take the expectation of our estimator then the expression reduces to \cref{eq:bias-identity} which shows that the bias is reduced to $\gamma^k$.
Here, we cheated slightly by assuming that the bias is the same for all the positions $i \in [\usize]$.
This is of course not true, but actual calculation reduces to $\usize$ instances of \cref{eq:bias-identity}.

The more delicate part of the analysis of our estimator is the bound on the variance.
The main difficulty lies in rewriting $\estimator_k - \E[\estimator_k]$ into something manageable.
We note that we can write our estimator $\estimator_k$ as
\begin{align*}
  \estimator_k = \sum_{i \in [\usize]} x_i \sum_{h = 1}^k (-1)^{h+1} \frac{\binom k h}{\binom m h}  \frac{\binom{Y_i}{h}}{(\calP(i))^h} \; .  
\end{align*}
We can express the frequency vector $Y$ by random variables as follows: $Y_i = \sum_{j \in [\ssize]} [X_j = i]$. 
(Here, we use $[X_j=i]$ to denote the indicator random variable of the event $X_j=i$.)
This allows us to see that $\binom{Y_i}{h} = \sum_{\substack{I \subseteq [\ssize]\\|I| = h}} \prod_{j \in I} [X_j = i]$.
If we now define the polynomials $P_i$ by 
\begin{align*}
  P_i(\beta_1, \ldots, \beta_{\ssize}) = \sum_{h = 1}^k (-1)^{h+1} \frac{\binom k h}{\binom m h}  \frac{1}{(\calP(i))^h} \sum_{\substack{I \subseteq [\ssize]\\|I| = h}} \prod_{j \in I} \beta_j
  \; ,
\end{align*}
then we can write our estimator as $\estimator_k = \sum_{i \in [\usize]} x_i P_i([X_1 = i], \ldots, [X_m = i])$.
We observe that $P$ has degree 1 in each variable so $\E[P_i([X_1 = i], \ldots, [X_m = i])] = P_i(\calQ(i), \ldots, \calQ(i))$.
Furthermore, it seems reasonable that there should exist polynomials $Q_i$ which have degree 1 in each variable and satisfy $P_i([X_1 = i], \ldots, [X_m = i]) - P_i(\calQ(i), \ldots, \calQ(i)) = Q_i([X_1 = i] - \calQ(i), \ldots, [X_m = i] - \calQ(i))$.
This will help us in understanding the variance of our final estimator $\zeta_k$ by decomposing into variances of the simpler events $[X_1=i]$.
We show that $Q_i$ exist and are defined as follows
\begin{align*}
  Q_i(\beta_1, \ldots, \beta_m) = (-1)^{k + 1} \sum_{h = 1}^k \frac{\binom k h}{\binom m h}  \frac{\gamma_i^{k - h}}{(\calP(i))^h} \sum_{\substack{I \subseteq [\ssize]\\|I| = h}} \prod_{j \in I} \beta_j
\end{align*}
So we get that $\estimator_k - \E[\estimator_k] = \sum_{i \in [\usize]}x_i Q_i([X_1 = i] - \calQ(i), \ldots, [X_m = i] - \calQ(i))$.
Since $[X_1 = i] - \calQ(i)$ is a zero mean variable and $Q_i$ has degree 1 in each variable, it becomes easy to calculate the variance.

A technical detail that we have not touched upon yet is that if $k \ge 2$ then $\Var[\estimator_k]$ becomes very large if $\mu^2 \gg \Var[\htest]$.
The reason is that $\Var[\estimator_k]$ depends on $\sum_{i \in [\usize]} \calP(i) \left(\frac{x_i}{\calP(i)}\right)^2 = \Var[\htest] + \mu^2$.
Thus, if $\mu^2 \gg \Var[\htest]$ then our variance becomes much larger which is a problem.
Now imagine that we already have an estimate, $W$, of $\mu$.
We then set $\bar{x}_i = x_i - \calP(i)W$ and make a new estimator $\bar{\estimator}_k$ that uses $\bar{x}_i$ instead of $x_i$.
Then $\bar{\estimator}_k$ will estimate $\mu - W$ so $\bar{\estimator}_k + W$ will be an estimator of $\mu$.
The trick is that $\Var[\bar{\estimator}_k]$ depends on $\Var[\htest] + (\mu - W)^2$ so if $W = \mu + O(\sqrt{\Var[\htest]})$ then we will have control over the variance.
We can get such estimate, $W$, by using our estimator for $k = 1$ where there are no issues with variance (i.e., the standard deviation of $W - \mu$ only depends on $\sqrt{\Var[\htest]}$ rather than on $\mu$).

\paragraph{Lower Bound.}

At a high level, our strategy is to define a reduction from the problem of distinguishing two distributions $D_1$ and $D_2$ to the problem of estimating $\mu = \sum_i x_i$. More concretely, for each fixed positive integer $k$, we define distributions $D_1$ and $D_2$ with support sizes $n_1$ and $n_2$, respectively, such that:
\begin{enumerate}
\item $n_1 = (1 + (\Theta(\gamma))^k) n_2$,
\item $D_1$ and $D_2$ are both pointwise $\gamma$-close to uniform, and
\item $D_1$ and $D_2$ have identical $p\th$ frequency moments for $p = 1, 2, \ldots, k$.\footnote{Recall that the $p\th$ frequency moment of a distribution $D$ is $\sum_{x} (D(x))^p$.}
\end{enumerate}
We describe a reduction showing that for $\usize = n_1 + n_2$, $\calP$ uniform over $[\usize]$, and $x_i \in \set{0, 1}$ for all $i \in [\usize]$, any algorithm that distinguishes $\mu = n_1$ from $\mu = n_2$ can also distinguish $D_1$ from $D_2$. We then apply a framework of by Raskhodnikova et al.~\cite{raskhodnikova2009strong}, which implies that distinguishing any two distributions whose first $k$ frequency moments are equal requires  $\Omega(n^{1-1/(k+1)})$ samples.
One difference between our setting and that of~\cite{raskhodnikova2009strong} is that in our setting, for each sample $i$ we are also given $x_i$, whereas~\cite{raskhodnikova2009strong} focuses on support size estimation. However, since the $x_i$'s are $0$ or $1$-valued, estimating the sum requires us to estimate either the number of $x_i$'s which are $0$ or the number of $x_i$'s which are $1$.
We can apply this observation to reduce the problem to finding two nearly uniform distributions (i.e., both are pointwise $\gamma$-close to uniform) that differ in support size by a multiplicative $1 \pm \Theta(\gamma)^k$ factor, yet match in the first $k$ moments.

To do this, we use a combinatorial construction that is inspired by a related lower bound in~\cite{eden2021-learning}. For simplicity, we assume that the probability of sampling any fixed item lies in
$\{\frac{1}{n}, \frac{1+\gamma}{n}, \dots, \frac{1+k \gamma}{n}\}$ (while these distributions would only be $k \cdot \gamma$-close to uniform, we can replace $\gamma$ with $\gamma/k$). For the distribution $D_1$, we assume the number of elements with probability $\frac{1+i \cdot \gamma}{n}$ is $a_i$, and for the distribution $D_2$ the number of elements with probability $\frac{1 + i \cdot \gamma}{n}$ is $b_i$. Then, our goal is for the support sizes of $D_1$ and $D_2$ (which are $\sum a_i, \sum b_i$, respectively) to differ significantly but for the first $k$ moments to match. This means $\sum a_i$ and $\sum b_i$ differ significantly, but $\sum a_i (1 + i \cdot \gamma)^\ell = \sum b_i (1 + i \cdot \gamma)^\ell$. If we define $c_i := a_i-b_i$ and think of $(1 + i \cdot \gamma)^\ell$ as a degree $\ell$ polynomial $P(i),$ we want $\sum c_i P(i) = 0$ but $\sum c_i$ to be large (roughly $\gamma^k \cdot n$). Finally, we need to make sure that $\sum a_i, \sum b_i$ are both $\Theta(n)$.

To determine the values of each $c_i$, we utilize the observation that for any polynomial $P(x)$ of degree less than $k$, the successive differences, i.e., $P(x)-P(x-1)$ is a polynomial of degree less than $k-1$.
We can repeatedly take successive differences $k$ times to get a linear combination of $P(0), P(1), \dots, P(k)$ that equals $0$ for any polynomial $P$ of degree less than $k$. We can therefore set $c_i := a_i-b_i$ to be these linear coefficients. Unfortunately, we will have $\sum c_i = 0$ as well. Instead, we replace $c_i$ with $c_i' := c_i/(1 + i \gamma)$, so that $\sum c_i' \cdot (1 + i \cdot \gamma)^\ell = \sum c_i \cdot (1 + i \cdot \gamma)^{\ell-1} = 0$ for all $1 \le \ell \le k$. However, $\sum c_i' = \sum c_i/(1+i \cdot \gamma)$ is not expressible as $\sum c_i P(i)$ for a polynomial $P$ of low degree, and will in fact be nonzero, which is exactly what we want. We scale the $c_i'$ terms so that $\sum c_i' = \gamma^k \cdot n.$ If we write $a_i = \max(c_i', 0)$ and $b_i = \max(-c_i', 0)$, then every $a_i$ and $b_i$ is nonnegative but $a_i-b_i = c_i'$. One can show via some careful combinatorics that after this scaling, $\sum a_i$ and $\sum b_i$ are both $\Theta(n)$, as desired.



\section{Sum Estimation}
\label{sec:sum-estimation}

We now give the algorithm for the sum estimation problem. We then state our main result (\Cref{thm:estimate-sum}) as well as the special case for estimating the sum of $0$-$1$ values (\Cref{cor:zero-one-sum}) that we discussed in the introduction. We then state and prove \Cref{lem:estimate-sum} which we then use to prove \Cref{thm:estimate-sum}.




\begin{algorithm}
    $(X_j)_{j \in [\ssize]} \leftarrow$ take $m$ samples from the distribution $\calQ$\\
    For every $i \in [\usize]$, let $Y_i$ denote the number of times value $i$ was sampled $\;\;\;\;$ ($Y_i=\sum_{j \in [\ssize]} [X_j = i]$ ) \\
    For every $h \in [k]$, let $\xi_h$ be the $h$-wise collision estimator $\qquad\qquad\qquad\qquad\;$ \big($\xi_h = \frac{1}{\binom \ssize h} \sum_{i = 1}^{\usize} \binom{Y_i}{h} \frac{x_i - \calP(i)W}{(\calP(i))^h} $ \big)\\
    \Return{$W + \sum_{h = 1}^{k} (-1)^{h + 1} \binom{k}{h} \xi_h$}

    \caption{$EstimateSum(m, k, W)$}
\end{algorithm}

\begin{algorithm}
    $W \leftarrow EstimateSum(t, 1, 0)$\\
    \Return{$EstimateSum(m, k, W)$}
    
    \caption{$ImprovedEstimateSum(m, t, k)$}
    \label{alg:improved-estimate}
\end{algorithm}

We now state our main theorem. Note that this implies, by a simple substitution, the simpler version mentioned in the introduction as \Cref{thm:intro-estimate-sum}.

\begin{theorem}
  \label{thm:estimate-sum}
    Define $\esize = \max_i 1 / \calP(i)$.
    Suppose $\calQ$ is pointwise $\gamma$-close to $\calP$, and let $\Varht$ denote the variance of the Hansen-Hurwitz estimator (Equation~\ref{eqn:ht-var}). For $\eps_1,\eps_2 > 0$, define $k = \lceil (\lg \eps_1)/\lg \gamma)\rceil$. Then using
    \[
    m = O\left(\sqrt[k]{n^{k-1} \eps_2^{-2} \Varht} \right)
    \]
     independent samples from $\calQ$, with probability at least $2/3$, Algorithm~\ref{alg:improved-estimate} produces an estimate $\hat\mu$ of $\mu = \sum_i x_i$ 
     with absolute error
     \[
    \abs{\mu - \hat\mu} \leq \eps_1(1 + \gamma) \E_{X \sim \calP}[|\calP(X)^{-1}x_X - \mu|] + \eps_2
    \]
\end{theorem}
This theorem implies in a straightforward way a solution to the problem of estimating the sum of $0$-$1$ values that we discussed in the introduction.
\begin{corollary}\label{cor:zeros_and_ones}
    Suppose $\calQ$ is pointwise $\gamma$-close to the uniform distribution over $[\usize]$ and $x_i \in \set{0, 1}$ for every $i \in [\usize]$. For any $\eps > 0$ define $k = \lceil (\log \eps) / \log \gamma \rceil$. Then $m = O(n^{1 - 1/k} \eps^{-2/k})$ samples are sufficient to obtain an estimate of $\mu = \sum_{i} x_i$ with additive error $\eps (\mu + \sqrt{\mu\usize})$ with probability $2/3$. 
  \end{corollary}

Before we can prove our main result, we need the following lemma.
\begin{lemma}[Analysis of $EstimateSum(m, k, 0)$]
    \label{lem:estimate-sum}
    Define $\esize = \max_i 1 / \calP(i)$. Let $(x_i)_{i \in [\usize]}$ be a sequence of numbers and define $\mu = \sum_{i \in [\usize]} x_i$.
    Let $\calP$ be a probability distribution over $[\usize]$, and let $\calQ$ be another probability distribution over $[\usize]$ that is pointwise $\gamma$-close to $\calP$. 

    Consider a sequence $(X_j)_{j \in [\ssize]}$ of independent random variables both with distribution $\calQ$.
    Define $Y_i = \sum_{j \in [\ssize]} [X_j = i]$ for every $i \in [\usize]$. 
    Define
    \begin{align*}
        \estimator_k = \sum_{h = 1}^{k} (-1)^{h + 1} \frac{\binom{k}{h}}{\binom{\ssize}{h}} \sum_{i \in [\usize]} \binom{Y_i}{h} \calP(i)^{-h} x_i
            \; .
    \end{align*}
    Then for $k \ge 2$, we get that
    \begin{align}
        |\E[\estimator_k] - \mu| &\le \gamma^k \sum_{i \in [\usize]} |x_i| \label{eq:expectation} \\
        \Var[\estimator_k] &\le  \max\!\left\{2 (1 + \gamma) \gamma^{2k - 2} k^2 \frac{\sum_{i \in [\usize]} \calP(i)^{-1} x_i^2 }{\ssize}, 2^{k} (1 + \gamma)^k  k^{3k} \frac{\esize^{k - 1} \sum_{i \in [\usize]} \calP(i)^{-1} x_i^2 }{m^k} \right\} \label{eq:variance}
            \; ,
    \end{align}
    and for $k = 1$, we get that
    \begin{align}
        |\E[\estimator_1] - \mu| &\le \gamma \sum_{i \in [\usize]} |x_i - \calP(i) \mu| \label{eq:expectation-k-1} \\
        \Var[\estimator_1] &\le (1 + \gamma) \frac{\sum_{i \in [\usize]} \calP(i)^{-1} (x_i - \calP(i) \mu)^2 }{\ssize} \label{eq:variance-k-1}
    \end{align}


\end{lemma}
\begin{proof}
    For each $i \in [\usize]$ we define $\gamma_i$ such that $\calQ(i) = (1 + \gamma_i)\calP(i)$.
    Since we know that $\calQ$ is pointwise $\gamma$-close to $\calP$ then $|\gamma_i| \le \gamma$,
    and since both $\calQ$ and $\calP$ are probability distributions then $\sum_{i \in [\usize]} \gamma_i \calP(i) = 0$.

    We start by proving the bounds on the expectation, i.e., \cref{eq:expectation} and \cref{eq:expectation-k-1}.
    \begin{align*}
        \E[\estimator_k]
            &= \E\left[\sum_{h = 1}^{k} (-1)^{h + 1} \frac{\binom{k}{h}}{\binom{\ssize}{h}} \sum_{i \in [\usize]} \binom{Y_i}{h} \calP(i)^{-h} x_i\right]
            \\&= \sum_{h = 1}^{k} (-1)^{h + 1} \frac{\binom{k}{h}}{\binom{\ssize}{h}} \sum_{i \in [\usize]} \E\left[\binom{Y_i}{h}\right] \calP(i)^{-h} x_i  
        \; .
    \end{align*}
    We use the fact that $Y_i = \sum_{j \in [\ssize]} [X_j = i]$ is a sum of 0-1 variables so $\binom{Y_i}{h} = \sum_{I \subseteq [\ssize] : |I| = h} \prod_{j \in I} [X_j = i]$.
    This implies that $\E\left[\binom{Y_i}{h}\right] = \binom{\ssize}{h} \calQ(i)^h = \binom{\ssize}{h} (1 + \gamma_i)^h \calP(i)^h$.
    Plugging this in, we get that
    \begin{align*}
        \E[\estimator_k]
            &= \sum_{h = 1}^{k} (-1)^{h + 1} \frac{\binom{k}{h}}{\binom{\ssize}{h}} \sum_{i \in [\usize]} \binom{\ssize}{h} (1 + \gamma_i)^h \calP(i)^h \calP(i)^{-h} x_i 
            \\&= \sum_{h = 1}^{k} (-1)^{h + 1} \binom{k}{h} \sum_{i \in [\usize]} (1 + \gamma_i)^{h} x_i
            \\&= \sum_{i \in [\usize]} x_i \sum_{h = 1}^{k} (-1)^{h + 1} \binom{k}{h} (1 + \gamma_i)^{h}
            \\&= \sum_{i \in [\usize]} x_i \left(1 + \sum_{h = 0}^{k} (-1)^{h + 1} \binom{k}{h} (1 + \gamma_i)^{h} \right)
            \\&= \sum_{i \in [\usize]} x_i \left(1 - (1 - (1 + \gamma_i))^k \right)
            \\&= \sum_{i \in [\usize]} x_i \left(1 + (-1)^{k + 1} \gamma_i^k \right)
    \end{align*}
    If $k \ge 2$ then we easily see that
    \begin{align*}
        |\E[\estimator_k] - \mu|
            = |\sum_{i \in [\usize]} x_i - \sum_{i \in [\usize]} x_i (1 + (-1)^{k + 1} \gamma_i^k ) |
            = |\sum_{i \in [\usize]} \gamma_i^k x_i|
            \le \gamma^k \sum_{i \in [\usize]} |x_i|
    \end{align*}
    For $k = 1$ we will exploit that $\sum_{i \in [\usize]} \gamma_i \calP(i) = 0$.
    \begin{align*}
        |\E[\estimator_1] - \mu|
            &= |\sum_{i \in [\usize]} \gamma_i x_i|
            = |\sum_{i \in [\usize]} \gamma_i (\calP(i)\mu + (x_i - \calP(i)\mu))|
            \\&= |\sum_{i \in [\usize]} \gamma_i (x_i - \calP(i)\mu)|
            \le \gamma \sum_{i \in [\usize]} |x_i - \calP(i)\mu|
    \end{align*}

    Now we will focus on bounding the variance.
    First we prove \cref{eq:variance-k-1}.
    \begin{align*}
        \Var[\estimator_1]
            &= \E\left[\left(\frac{1}{m} \sum_{j \in [\ssize]}  \sum_{i \in [\usize]} x_i ([X_j = i]\calP(i)^{-1} - (1 + \gamma_i))  \right)^2 \right]
            \\&= \frac{1}{m^2} \sum_{j \in [\ssize]} \E\left[\left( \sum_{i \in [\usize]} x_i ([X_j = i]\calP(i)^{-1} - (1 + \gamma_i))  \right)^2 \right]
    \end{align*}
    We will argue that $\sum_{i \in [\usize]} x_i ([X_j = i]\calP(i)^{-1} - (1 + \gamma_i)) = \sum_{i \in [\usize]} (x_i - \calP(i)\mu) ([X_j = i]\calP(i)^{-1} - (1 + \gamma_i))$.
    \begin{align*}
        \sum_{i \in [\usize]} x_i ([X_j = i]\calP(i)^{-1} - (1 + \gamma_i))
            &= \sum_{i \in [\usize]} (x_i - \calP(i)\mu + \calP(i)\mu) ([X_j = i]\calP(i)^{-1} - (1 + \gamma_i))
            \\&= \sum_{i \in [\usize]} (x_i - \calP(i)\mu) ([X_j = i]\calP(i)^{-1} - (1 + \gamma_i))
                + \mu \sum_{u \in [\usize]} \calP(i)\gamma_i
    \end{align*}
    Now since $\sum_{i \in [\usize]} \gamma_i \calP(i) = 0$ it follows.
    We can now bound the variance.
    \begin{align*}
        \Var[\estimator_1]
            &= \frac{1}{m^2} \sum_{j \in [\ssize]} \E\left[\left( \sum_{i \in [\usize]} (x_i - \calP(i)\mu) ([X_j = i]\calP(i)^{-1} - (1 + \gamma_i))  \right)^2 \right]
            \\&\le \frac{1}{m^2} \sum_{j \in [\ssize]} \E\left[\left( \sum_{i \in [\usize]} \calP(i)^{-1} (x_i - \calP(i)\mu) [X_j = i]  \right)^2 \right]
            \\&= \frac{1}{m} \sum_{i \in [\usize]} \frac{\calQ(i)}{\calP(i)^2} (x_i - \calP(i)\mu)^2
            \\&\le \frac{1 + \gamma}{m} \sum_{i \in [\usize]} \calP(i)^{-1} (x_i - \calP(i)\mu)^2
    \end{align*}

    Now we just need to focus on the case of $k \ge 2$ and prove \cref{eq:variance}.
    For this we need the following lemma for which we defer the proof till \Cref{sec:proof-technical}.
    \begin{lemma}\label{lem:expression-mean-zero}
        For all sequences of numbers $(\beta_j)_{j \in [\ssize]}$ and all $\alpha$ the following identity holds:
        \begin{align*}
            \sum_{\substack{I \subseteq [\ssize]\\0 < |I| \le k}} (-1)^{|I| + 1} \frac{\binom{k}{|I|}}{\binom{\ssize}{|I|}} \left(\prod_{j \in I} \beta_j - (1 + \alpha)^h\right)
                = (-1)^{k + 1} \sum_{\substack{I \subseteq [\ssize]\\0 < |I| \le k}} \frac{\binom{k}{|I|}}{\binom{m}{|I|}} \alpha^{k - |I|} \prod_{j \in I} (\beta_j - (1 + \alpha))
        \end{align*}
    \end{lemma}

    The idea is to use \Cref{lem:expression-mean-zero} to prove that
    \begin{align}\label{eq:rewriting-for-variance}
        \estimator_k - \E[\estimator_k]  = \sum_{\substack{I \subseteq [\ssize]\\0 < |I| \le k}} \frac{\binom{m - |I|}{k - |I|}}{\binom{m}{k}} \alpha^{k - |I|}  \sum_{i \in [\usize]} x_i \prod_{j \in I} (\calP(i)^{-1}[X_j = i] - (1 + \gamma_i)) \; .
    \end{align}
    First we use that $\E[\estimator_k] = \sum_{h = 1}^{k} (-1)^h \frac{\binom{k}{h}}{\binom{m}{h}} \sum_{i \in [\usize]} \calP(i)^{-h} x_i \E\left[\binom{Y_i}{h}\right]$ which allow us to rewrite $\estimator_k - \E[\estimator_k]$.
    \begin{align*}
        \estimator_k - \E[\estimator_k] = \sum_{h = 1}^{k} (-1)^h \frac{\binom{k}{h}}{\binom{m}{h}} \sum_{i \in [\usize]} \calP(i)^{-h} x_i \left( \binom{Y_i}{h} - \E\left[\binom{Y_i}{h}\right] \right)
    \end{align*}
    We now again use that $Y_i = \sum_{j \in [\ssize]} [X_j = i]$ is a sum of 0-1 variables so $\binom{Y_i}{h} = \sum_{I \subseteq [\ssize] : |I| = h} \prod_{j \in I} [X_j = i]$ and $\E\left[\binom{Y_i}{h}\right] = \binom{\ssize}{h} (1 + \gamma_i)^h \calP(i)^h$.
    \begin{align*}
        &\sum_{h = 1}^{k} (-1)^h \frac{\binom{k}{h}}{\binom{m}{h}} \sum_{i \in [\usize]} \calP(i)^{-h} x_i \left( \binom{Y_i}{h} - \E\left[\binom{Y_i}{h}\right] \right)
            \\&= \sum_{h = 1}^{k} (-1)^h \frac{\binom{k}{h}}{\binom{m}{h}} \sum_{i \in [\usize]} \calP(i)^{-h} x_i \sum_{\substack{I \subseteq [m]\\|I| = h}} \left(\prod_{j \in I}[X_j = i] - (1 + \gamma_i)^h \calP(i)^h \right)
            \\&= \sum_{i \in [\usize]} x_i \sum_{h = 1}^{k} \sum_{\substack{I \subseteq [m]\\|I| = h}} (-1)^h \frac{\binom{k}{h}}{\binom{m}{h}} \calP(i)^{-h}   \left(\prod_{j \in I}[X_j = i] - (1 + \gamma_i)^h \calP(i)^h \right)
            \\&= \sum_{i \in [\usize]} x_i \sum_{\substack{I \subseteq [m]\\0 < |I| \le k}} (-1)^{|I| + 1} \frac{\binom{k}{|I|}}{\binom{m}{|I|}} \calP(i)^{-h}   \left(\prod_{j \in I}[X_j = i] - (1 + \gamma_i)^{|I|} \calP(i)^{|I|} \right)
            \\&= \sum_{i \in [\usize]} x_i \sum_{\substack{I \subseteq [m]\\0 < |I| \le k}} (-1)^{|I| + 1} \frac{\binom{k}{|I|}}{\binom{m}{|I|}} \left(\prod_{j \in I} \calP(i)^{-1} [X_j = i] - (1 + \gamma_i)^{|I|} \right)
    \end{align*}
    For each $i \in [\usize]$, the expression $\sum_{\substack{I \subseteq [m]\\0 < |I| \le k}} (-1)^{|I| + 1} \frac{\binom{k}{|I|}}{\binom{m}{|I|}} \left(\prod_{j \in I} \calP(i)^{-1} [X_j = i] - (1 + \gamma_i)^{|I|} \right)$ is of the form of \Cref{lem:expression-mean-zero}.
    So applying that $\usize$ times we get that
    \begin{align*}
        &\sum_{i \in [\usize]} x_i \sum_{\substack{I \subseteq [m]\\0 < |I| \le k}} (-1)^{|I| + 1} \frac{\binom{k}{|I|}}{\binom{m}{|I|}} \left(\prod_{j \in I} \calP(i)^{-1} [X_j = i] - (1 + \gamma_i)^{|I|} \right)
            \\&= \sum_{i \in [\usize]} x_i (-1)^{k + 1} \sum_{\substack{I \subseteq [m]\\0 < |I| \le k}} \frac{\binom{k}{|I|}}{\binom{m}{|I|}} \gamma_i^{k - |I|} \prod_{j \in I} \left(\calP(i)^{-1} [X_j = i] - (1 + \gamma_i) \right)
            \\&= (-1)^{k + 1} \sum_{\substack{I \subseteq [m]\\0 < |I| \le k}} \frac{\binom{k}{|I|}}{\binom{m}{|I|}} \sum_{i \in [\usize]}  \gamma_i^{k - |I|}  x_i \prod_{j \in I} \left(\calP(i)^{-1} [X_j = i] - (1 + \gamma_i) \right)
    \end{align*}
    This shows that \cref{eq:rewriting-for-variance} is true.

    Now let $I_1, I_2 \subseteq [m]$ be two different set of indices with $0 < |I_1|, |I_2| \le k$.
    Without loss of generality, we can assume that there exists $h \in I_1 \setminus I_2$.
    \begin{align*}
        T = (\sum_{i \in [\usize]} \gamma_i^{k - |I_1|} x_i \prod_{j \in I_1} (\calP(i)^{-1}[X_j = i] - (1 + \gamma_i))) (\sum_{i \in [\usize]} \gamma_i^{k - |I_2|} x_i \prod_{j \in I_2} (\calP(i)^{-1}[X_j = i] - (1 + \gamma_i)))
    \end{align*}
    Note that if we multiply this expression out then every term will contain a factor of the form $(\calP(i)^{-1}[X_h = s] - (1 + \gamma_i))$ for some $s \in [\usize]$ and where all the other factors are independent of $X_h$.
    Since $\E[(\calP(i)^{-1}[X_h = s] - (1 + \gamma_i))] = 0$ we get that $\E[T] = 0$.
    This implies that
    \begin{align*}
        \E[(\estimator_k - \E[\estimator_k])^2]
            = \sum_{\substack{I \subseteq [\ssize]\\0 < |I| \le k}} \left(\frac{\binom{k}{|I|}}{\binom{m}{|I|}} \right)^2  \E[(\sum_{i \in [\usize]} \gamma_i^{k - |I|} x_i \prod_{j \in I} (\calP(i)^{-1}[X_j = i] - (1 + \gamma_i)))^2]
    \end{align*}

    Let $I \subseteq [\usize]$ be fixed then 
    \begin{align*}
        \E[(\sum_{i \in [\usize]} \gamma_{i}^{k - |I|} x_i \prod_{j \in I} (\calP(i)^{-1}[X_j = i] - (1 + \gamma_i)))^2] 
            &\le \E[(\sum_{i \in [\usize]} \gamma_{i}^{k - |I|} x_i \prod_{j \in I} \calP(i)^{-1}[X_j = i])^2] 
            \\&= \sum_{i \in [\usize]} \gamma_{i}^{2k - 2|I|} x_i^2 \frac{\calQ(i)^{|I|}}{\calP(i)^{2|I|}}
            \le \gamma^{2k - 2|I|} \sum_{i \in [\usize]} x_i^2 \frac{(1 + \gamma_i)^{|I|}}{\calP(i)^{|I|}}
    \end{align*}
    Collecting it we get that
    \begin{align*}
        \E[(\estimator_k - \E[\estimator_k])^2]
            &\le \sum_{\substack{I \subseteq [\ssize]\\0 < |I| \le k}} \left(\frac{\binom{k}{|I|}}{\binom{m}{|I|}} \right)^2 \gamma^{2k - 2|I|}  \sum_{i \in [\usize]} x_i^2 \frac{(1 + \gamma_i)^{|I|}}{\calP(i)^{|I|}}
            \\&= \sum_{h = 1}^{k} \frac{\binom{k}{h}^2}{\binom{m}{h}} \gamma^{2k - 2h}  \sum_{i \in [\usize]} x_i^2 \frac{(1 + \gamma_i)^{h}}{\calP(i)^{h}}
            \\&\le \sum_{h = 1}^{k} \frac{\binom{k}{h}^2}{\binom{m}{h}} \gamma^{2k - 2h} (1 + \gamma)^h \sum_{i \in [\usize]} \frac{x_i^2}{\calP(i)^{h}}
            \\&\le \sum_{h = 1}^{k} \frac{k^{2h}}{\binom{m}{h}} \gamma^{2k - 2h} (1 + \gamma)^h \sum_{i \in [\usize]} \frac{x_i^2}{\calP(i)^{h}}
            \\&\le \max_{h = 1}^{k} \frac{2^h k^{2h}}{\binom{m}{h}} \gamma^{2k - 2h} (1 + \gamma)^h \sum_{i \in [\usize]} \frac{x_i^2}{\calP(i)^{h}}
    \end{align*}

    Now we note that for all $i \in [\usize]$, the map $h \mapsto \frac{2^h k^{2h}}{\binom{m}{h}} \gamma^{2k - 2h} (1 + \gamma)^h \frac{x_i^2}{\calP(i)^{h}}$ is log-convex since each factor is log-convex.
    This implies that the map $h \mapsto \frac{2^h k^{2h}}{\binom{m}{h}} \gamma^{2k - 2h} (1 + \gamma)^h \sum_{i \in [\usize]} \frac{x_i^2}{\calP(i)^{h}}$ is convex and is thus maximized at the boundary.
    We then get that
    \begin{align*}
        \E[(\estimator_k - \E[\estimator_k])^2]
            \le \max\!\left\{2 (1 + \gamma) \gamma^{2k - 2} k^2 \frac{\sum_{i \in [\usize]} \calP(i)^{-1} x_i^2 }{\ssize}, 2^{k} (1 + \gamma)^k  k^{3k} \frac{\esize^{k - 1} \sum_{i \in [\usize]} \calP(i)^{-1} x_i^2 }{m^k} \right\}
    \end{align*}
    which finishes the proof of \cref{eq:variance}.
\end{proof}

We are now ready to prove the main theorem.
\begin{proof}[Proof of \Cref{thm:estimate-sum}]
    Let $W$ be an estimate of $\mu$ using $EstimateSum(t, 1, 0)$ where $t = O(1 + \gamma^{2k} \eps_2^{-2} \Varht )$.
    Using \Cref{lem:estimate-sum} we get an estimate of the bias and the variance.
    Now by Chebyshev's inequality, we easily get that $|W - \mu| \leq \gamma\sum_{i \in [\usize]}\calP(i)|\calP(i)^{-1}x_i - \mu| + \max\set{\eps_2/(2\gamma^k), \sqrt{\sum_{i \in [\usize]} \calP(i)(\calP(i)^{-1}x_i - \mu)^2 }} $ with probability 5/6.
    We note that $\sum_{i \in [\usize]}\calP(i)|\calP(i)^{-1}x_i - \mu| \le  \sqrt{\sum_{i \in [\usize]} \calP(i)(\calP(i)^{-1}x_i - \mu)^2 }$ since $p$-norms are increasing.
    So we also get that $|W - \mu| \le O(\sqrt{\sum_{i \in [\usize]} \calP(i)(\calP(i)^{-1}x_i - \mu)^2 })$

    Now we calculate $\estimator$ using $EstimateSum(m, k, W)$ with $m = O\left(\sqrt[k]{n^{k-1} \eps_2^{-2} \Varht} \right)$ so $\estimator$ corresponds to $ImprovedEstimateSum(m, t, k)$.
    We now note that by \Cref{lem:estimate-sum},
    \begin{align*}
        |(\E[\estimator] - W) - (\mu - W)| 
            &\le \gamma^k \sum_{i \in [\usize]} \calP(i)|\calP(i)^{-1}x_i - W|
            \le \gamma^k|W - \mu| + \gamma^k \sum_{i \in [\usize]} \calP(i)|\calP(i)^{-1}x_i - \mu| \\
        \Var[\estimator]
            &\le \max\!\bigg\{2 (1 + \gamma) \gamma^{2k - 2} k^2 \frac{\sum_{i \in [\usize]} \calP(i) (\calP(i)^{-1} x_i - W)^2 }{\ssize}, 
            \\& \qquad\qquad2^{k} (1 + \gamma)^k  k^{3k} \frac{\esize^{k - 1} \sum_{i \in [\usize]} \calP(i) (\calP(i)^{-1} x_i - W)^2 }{m^k} \bigg\}
            \\&= \max\!\bigg\{2 (1 + \gamma) \gamma^{2k - 2} k^2 \frac{\sum_{i \in [\usize]} \calP(i) (\calP(i)^{-1} x_i - \mu)^2 + (W - \mu)^2 }{\ssize}, 
            \\& \qquad\qquad2^{k} (1 + \gamma)^k  k^{3k} \frac{\esize^{k - 1} \left( \sum_{i \in [\usize]} \calP(i) (\calP(i)^{-1} x_i - W)^2 + (W - \mu)^2 \right) }{m^k} \bigg\}
    \end{align*}

    Assuming that $|W - \mu| = \gamma\sum_{i \in [\usize]}\calP(i)|\calP(i)^{-1}x_i - \mu| + \max\set{\eps_1/(2\gamma^k), \sqrt{\sum_{i \in [\usize]} \calP(i)(\calP(i)^{-1}x_i - \mu)^2 }}$.
    We get that $|(\E[\estimator] - W) - (\mu - W)| \le \gamma^k(1 + \gamma)\sum_{i \in [\usize]}\calP(i)|\calP(i)^{-1}x_i - \mu| + \eps_2/2$.
    Using that $|W - \mu| \le O(\sqrt{\sum_{i \in [\usize]} \calP(i)(\calP(i)^{-1}x_i - \mu)^2 })$ we can then use Chebyshev's inequality to conclude that with probability 5/6
    \begin{align*}
        |\estimator - \E[\estimator]| \le \eps_2/2
    \end{align*}
    So all in all, with probability at least 2/3 we that $|\estimator - \mu| \le \eps_1(1 + \gamma)\sum_{i \in [\usize]}\calP(i)|\calP(i)^{-1}x_i - \mu| + \eps_2$.
\end{proof}

\section{Lower Bound}
\label{sec:lower-bound}

In this section, we prove that for a range of parameters---specifically in the setting of Corollary~\ref{cor:zeros_and_ones}---the sample complexity of our sum estimation algorithm (Algorithm~\ref{alg:improved-estimate}) is asymptotically tight.

\begin{theorem}\label{thm:lb}
  Let $k$ be a positive integer and let $\eps < \gamma < 1$ be positive numbers. Suppose $A$ is an algorithm such that for any $v \in \set{0, 1}^\usize$ and any distribution $\calP$ over $[\usize]$ that is pointwise $\gamma$-close to uniform, $A$ uses samples from $\calP$ and returns an estimate of $\lVert v \rVert_1 = \sum_i v_i$ within additive error $\eps \usize$ with probability $2/3$. Then there exists $c_k$ with $0 < c_k < 1$ such that if $\eps \leq c_k \gamma^k$ then $A$ requires $\Omega(\usize^{1-1/(k+1)})$ samples.
\end{theorem}

While at a first glance this result might seem contradictory to our upper bound (specifically, Corollary~\ref{cor:zero-one-sum}), it actually reveals the following interesting phenomena. 
Notice that in our upper bound, the complexity depends on $k=\lceil (\log \eps) / \log \gamma \rceil$, so that, e.g., for $\eps=\gamma^k$, the complexity is $O(\usize^{1-\frac{1}{k}})$. Once $\eps$ becomes slightly smaller, i.e., $\eps=c \gamma^k$ for $c$ satisfying $\gamma \leq c < 1$, the complexity of our algorithm abruptly jumps to $O(n^{1-\frac{1}{k+1}})$. The lower bound implies that this increase in complexity is unavoidable for sufficiently small $c$. That is, Theorem~\ref{thm:lb} states that there exists a (sufficiently small) constant $c_k$, such that indeed once $\eps = c_k \cdot \gamma^k$, the required number of samples is $\Omega(n^{1-\frac{1}{k+1}})$. Interestingly, for all $\eps$ satisfying $\gamma^{k+1} \leq \eps \leq c_k \gamma^k$, the asymptotic complexity of sum estimation is the same; the complexity only varies for $\eps$ satisfying $c_k \gamma^k \leq \eps \leq \gamma^k$. Our matching upper and lower bounds demonstrate that the sample complexity's non-uniform dependence on $\eps$ is not an artifact, but captures the true complexity of the problem (up to the dependency on $\gamma^{k/2}$ in the numerator of the upper bound). Note that if the conclusion of Theorem~\ref{thm:lb} held \emph{every} $c<1$, then this would capture the right dependency on $n$ for \emph{all} possible ranges of $\gamma$. Since we only prove the theorem for a sufficiently small $c_k$, it might be the case that for values $\eps$ that are not too much smaller than $\gamma^k$, the optimal dependency on $\usize$ is lower than our stated upper bound. Nonetheless, our upper and lower bounds match (up to constant factors) for all $\eps$ satisfying $\gamma^{k+1} \leq \eps \leq c_k \gamma^k$.

As described in Section~\ref{sec:technical-overview}, the main technical ingredient in the proof of the theorem is in describing two distributions $D_1$ and $D_2$ over ranges $[n_1], [n_2]$,  respectively, such that $D_1$ and $D_2$ are pointwise $\gamma$-close to uniform,  $n_1=(1+\Theta(\gamma)^k)n_2$, and $D_1$ and $D_2$ have matching frequency moments $1$ through $k$.
Given these distributions  we  rely on the framework for proving lower bounds by Raskhodnikova et al.~\cite{raskhodnikova2009strong}, which states that any uniform algorithm that distinguishes two random variables with matching frequency moments $1$ through $k$ must perform $\Omega(n^{1-1/(k+1)})$ many samples.

In order to simplify our construction and its analysis, we prove the lower bound for \emph{uniform algorithms}. Here, a uniform algorithm is an algorithm whose output depends only on the ``collision statistics'' of the samples---i.e., the number of collisions involving each sample, and not the identities of the samples themselves.

\begin{definition}[Uniform algorithm. Definition 3.2 in~\cite{raskhodnikova2009strong}]
  An algorithm is \emph{uniform} if it samples indices $i_1, \cdots, i_m$ independently with replacement and its output is uniquely determined by (i) the value of the items $x_{i_j}$ and (ii) the set of collisions, where two indices $j$ and $j'$ \emph{collide} if $i_{j} = i_{j'}$.
\end{definition}

In particular, a uniform algorithm's output does \emph{not} depend on the sampled indices themselves. The following lemma asserts that our restriction to uniform algorithms is without loss of generality.

\begin{lemma}[cf.~Theorem~11.12 in~\cite{goldreich2017introduction}]
  \label{lem:wlog-uniform}
  Suppose there exists an algorithm $A$ such that for any $v \in \set{0, 1}^\usize$ $A$ uses samples from $\calP$ and returns an estimate of $\lVert v \rVert_1 = \sum_i v_i$ within additive error $\eps \usize$ with probability $2/3$ using $s$ samples in expectation. Then there exists a uniform algorithm $A'$ that achieves the same approximation guarantee using $s$ samples in expectation.
\end{lemma}

The proof of Lemma~\ref{lem:wlog-uniform} is essentially the same as that of Goldreich's Theorem~11.12 in~\cite{goldreich2017introduction}. The key idea is that $\sum_i v_i$ is a ``label invariant'' in the sense that it is unaffected by any permutation of the indices of $v$. Thus, given an algorithm $A$ as in the hypothesis of Lemma~\ref{lem:wlog-uniform}, we can obtain uniform algorithm $A'$ with the same approximation guarantee by simply choosing a uniformly random permutation of the indices of $v$, then using the permuted indices of $v$ as inputs for $A$. See Theorem~11.12 in~\cite{goldreich2017introduction} for details.

Finally, our lower bound argument requires the following result that is a direct consequence of the work of Raskhodnikova et al.~\cite{raskhodnikova2009strong}.



\begin{theorem}[Consequence of Lemma~5.3 and~Corollary 5.7 from \cite{raskhodnikova2009strong}]
  \label{thm:freq-moments}
Let $D_1$ and $D_2$ be distributions over positive integers $b_1 < \ldots < b_{\ell}$, that have matching frequency moments $1$ through $k$.  Then for any uniform algorithm $A$ with sample complexity $s$ that distinguishes $D_1$ and $D_2$ with high constant probability, $s = \Omega(n^{1-1/(k+1)})$.
\end{theorem}

Our main argument for the lower bound applies Theorem~\ref{thm:freq-moments} in conjunction with a reduction from distinguishing $D_1$ and $D_2$ to sum estimation. The main technical ingredient is the following lemma, which asserts the existence of suitable distributions $D_1$ and $D_2$.

\begin{lemma}\label{lem:matching-freq-moments}
For every positive integer $k$ and sufficiently large integer $n$, there exist two distributions $D_1, D_2$ over $[n_1]$ and $[n_2]$ (respectively) satisfying $n _1 =(1+\Theta(\gamma)^k)n$, and $n_2 = n$ such that $p_{D_j}(i)\in (1\pm \gamma)\frac{1}{n}$ for $j\in\{1,2\}$ and the following holds. For all $\ell \in \{1, 2, \dots, k\}$, it holds that
\[
\sum_{i=1}^{n_1} (p_{D_1}(i))^\ell =\sum_{i=1}^{n_2} (p_{D_2}(i))^\ell.
\]
In particular, there exists an absolute constant $c_k$ such that for sufficiently large $n$, $n_2 \geq (1 + c_k \gamma^k) n_1$ and the above conclusion holds.
\end{lemma}

Before proving the lemma, we show how the lemma implies Theorem~\ref{thm:lb}.

\begin{proof}[Proof of Theorem~\ref{thm:lb}]
The theorem follows from Theorem~\ref{thm:freq-moments} together with Lemma~\ref{lem:matching-freq-moments}. Let $N = n_1+n_2$, where $n_1=(1+\Theta(\gamma)^k) n_2$ as in the conclusion of Lemma~\ref{lem:matching-freq-moments}, and consider distinguishing between two possible outcomes $\mathcal{O}_1$ and $\mathcal{O}_2$. 

In the first outcome $\mathcal{O}_1$, let $S = \{1, 2, \dots, n_1\} \subseteq [N]$ and $T = \{t_1, \dots, t_{n_2}\} := [N] \backslash S$. The distribution $\mathcal{Q}$ will be as follows. With exactly $1/2$ probability, we choose $S$: if so, we then choose a sample $i \sim D_1$, which will be in $[n_1]$, and output $i$. Otherwise, we choose $T$: we then choose a sample $i \sim D_2$, which will be in $[n_2]$, and then output $t_i$. Here, $D_1, D_2$ are the distributions from Lemma~\ref{lem:matching-freq-moments}. Finally, we let $v_i = 1$ if $i \in S$ and $v_i = 0$ if $i \in T$.

The second outcome $\mathcal{O}_2$ is similar but ``flipped''. Now, we let $S = \{1, 2,  \dots, n_2\} \subseteq [N]$, and $T = \{t_1, \dots, t_{n_1}\} := [N] \backslash S$. With exactly $1/2$ probability, we choose $S$: if so, we then choose a sample $i \sim D_2$, and output $i$. Otherwise, we choose $T$: we then choose a sample $i \sim D_1$, and then output $t_i$. Finally, we let $v_i = 1$ if $i \in S$ and $v_i = 0$ if $i \in T$.

Under $\mathcal{O}_1$, we have that $\sum v_i$ always equals $n_1$, whereas under $\mathcal{O}_2$, we have that $\sum v_i$ is always $n_2$. In addition, since both $n_1$ and $n_2$ are $\frac{N}{2} \cdot (1 \pm O(\gamma))$, and since $D_1$ and $D_2$ are $\gamma$-pointwise close to uniform, the distribution $\mathcal{Q}$ that we sample from in either case is $O(\gamma)$-pointwise close to uniform. So, we may assume that $\mathcal{P}$ is uniform over $[N]$ in either case.

Now, assume that there exists a uniform algorithm\footnote{Again, the assumption that $A$ is uniform is without loss of generality by Lemma~\ref{lem:wlog-uniform}. For our construction, however, the two scenarios $\mathcal{O}_1$ and $\mathcal{O}_2$ can be distinguished by a non-uniform algorithm using $O(\gamma^k)$ samples. Indeed, the two scenarios are distinguished by seeing any value $v_i$ with $n_2 < i \leq n_1$. Following the proof of Lemma~\ref{lem:wlog-uniform} (cf.~Theorem~11.12 in~\cite{goldreich2017introduction}), the scenarios are indistinguishable to even a non-uniform algorithm we we replace $S$ and $T$ with randomly chosen complementary subsets of $[N]$.} $A$ that draws samples $(i, v_i)$ either from outcome $\mathcal{O}_1$ or outcome $\mathcal{O}_2$ and with probability at least $2/3$, computes an estimate of $\sum_i v_i$ up to additive error $c_k \gamma^k \usize / 5$, where $c_k$ is as in the second conclusion of Lemma~\ref{lem:matching-freq-moments}. Observe that when the error bound on $A$ is satisfied (which occurs with probability at least $2/3$), $A$'s output distinguishes scenarios $\mathcal{O}_1$ and $\mathcal{O}_2$.

Finally, we observe that distinguishing $\mathcal{O}_1$ from $\mathcal{O}_2$ is sufficient to distinguish the distributions $D_1$ and $D_2$. Indeed, under scenario $\mathcal{O}_1$, the $1$-values and sampled from $D_1$, while the $0$-values are sampled from $D_2$, while the roles are reversed in $\mathcal{O}_2$. Thus, the output of $A$ suffices to distinguish $D_1$ and $D_2$. Since $A$ uses $s$ samples in expectation, Theorem~\ref{thm:freq-moments} and Lemma~\ref{lem:matching-freq-moments} imply that $s = \Omega(n^{1-1/(k+1)})$, as desired.
%
\end{proof}

We now conclude by proving our main technical lemma.

\begin{proof}[Proof of Lemma~\ref{lem:matching-freq-moments}]

    First, we note that 
\begin{align}
    \sum_{i = 0}^{k} (-1)^i {k \choose i} {i \choose r} &= \sum_{i = r}^{k} (-1)^i {k \choose i} {i \choose r} \\
    &= \sum_{i = r}^{k} (-1)^i \cdot \frac{k!}{(k-i)! (i-r)! r!} \\
    &= \sum_{i = r}^{k} (-1)^i \cdot {k \choose r} {k-r \choose i-r} \\
    &= (-1)^r {k \choose r} \cdot \sum_{j = 0}^{k-r} (-1)^j {k-r \choose j}.
\end{align}
    The last line follows by setting $j = i-r$. Now, note that the summation in the last line equals $(1-1)^{k-r} = 0$ if $k > r$, and equals $1$ if $k = r$. So, this means that $\sum_{i = 0}^{k} (-1)^i {k \choose i} {i \choose r} = 0$ for all $0 \le r < k$.
    
    Next, note that ${i \choose r} = \frac{i(i-1) \cdots (i-r+1)}{r!}$ for all integers $i \ge 0$. This is a degree-$r$ polynomial in $i$. From this observation, it is well-known that every degree at most $k-1$ polynomial in $i$ can be written as a linear combination of ${i \choose 0}, \dots, {i \choose k-1}$. Therefore, for any polynomial $P$ of degree at most $k-1$, $\sum_{i = 0}^{k} (-1)^i {k \choose i} \cdot P(i) = 0$.
    
    Now, we let the distribution $D_1$ have exactly a $\frac{{k \choose i}}{2^{k-1}}$ fraction of its mass consisting of items each with probability $\left(1 + \frac{\gamma \cdot i}{k}\right) \cdot \frac{1}{n_0}$, for each \emph{even} integer $0 \le i \le k$. Here, $n_0$ will be an integer chosen later. Note this means it must have $n_0 \cdot \frac{{k \choose i}}{2^{k-1} \cdot (1 + \frac{\gamma \cdot i}{k})}$ points with mass $\left(1 + \frac{\gamma \cdot i}{k}\right) \cdot \frac{1}{n_0}$ for each \emph{even} integer $0 \le i \le k$. Likewise, we let the distribution $D_2$ have exactly a $\frac{{k \choose i}}{2^{k-1}}$ fraction of its mass consisting of items each with probability $\left(1 + \frac{\gamma \cdot i}{k}\right) \cdot \frac{1}{n_0}$, for each \emph{odd} integer $0 \le i \le k$. Note that the total fraction of mass for both $D_1$ and $D_2$ is clearly $1$.
    
    First, we note that for any $1 \le \ell \le k$, 
\begin{align}
    &\hspace{0.5cm}\sum_{i = 1}^{n_1} (p_{D_1}(i))^{\ell} - \sum_{i = 1}^{n_1} (p_{D_2}(i))^{\ell} \\
    &= \sum_{\substack{i = 0 \\ i \text{ even}}}^{k} n_0 \cdot \frac{{k \choose i}}{2^{k-1} \cdot (1 + \frac{\gamma \cdot i}{k})} \cdot \left(\left(1 + \frac{\gamma \cdot i}{k}\right) \cdot \frac{1}{n_0}\right)^\ell - \sum_{\substack{i = 0 \\ i \text{ odd}}}^{k} n_0 \cdot \frac{{k \choose i}}{2^{k-1} \cdot (1 + \frac{\gamma \cdot i}{k})} \cdot \left(\left(1 + \frac{\gamma \cdot i}{k}\right) \cdot \frac{1}{n_0}\right)^\ell \\
    &= \frac{1}{n_0^{\ell-1} \cdot 2^{k-1}} \cdot \sum_{i = 0}^{k} (-1)^i {k \choose i} \cdot \left(1 + \frac{\gamma \cdot i}{k}\right)^{\ell-1}.
\end{align}
    By letting $P(i)$ be the polynomial $\left(1 + \frac{\gamma \cdot i}{k}\right)^{\ell-1}$, we have that $P(i)$ has degree at most $k-1$, so this equals $0$, as desired.
    
    Finally, we look at the difference $n_1-n_2$, i.e., the difference in support size between $D_1$ and $D_2$. This simply equals
\begin{align}
    \sum_{\substack{i = 0 \\ i \text{ even}}}^{k} n_0 \cdot \frac{{k \choose i}}{2^{k-1} \cdot (1 + \frac{\gamma \cdot i}{k})} - \sum_{\substack{i = 0 \\ i \text{ odd}}}^{k} n_0 \cdot \frac{{k \choose i}}{2^{k-1} \cdot (1 + \frac{\gamma \cdot i}{k})} &= \frac{n_0}{2^{k-1}} \cdot \sum_{i = 0}^{k} (-1)^i \cdot \frac{{k \choose i}}{1 + \frac{\gamma}{k} \cdot i}.
\end{align}

    We now inductively prove (by inducting on $k \ge 1$) that $\sum_{i = 0}^{k} (-1)^i \cdot \frac{{k \choose i}}{a + \gamma \cdot i} = \frac{k! \cdot \gamma^k}{a(a+\gamma) \cdots (a+k \gamma)}$ for any real numbers $a, \gamma$.
    For $k = 1$, we have that $\sum_{i = 0}^{k} (-1)^i \cdot \frac{{k \choose i}}{a + \gamma \cdot i} = \frac{1}{a} - \frac{1}{a+\gamma} = \frac{\gamma}{a(a+\gamma)}$. For general $k$, we can write 
\begin{align}
    \sum_{i = 0}^{k} (-1)^i \cdot \frac{{k \choose i}}{a + \gamma \cdot i} &= \sum_{i = 0}^{k} (-1)^i \cdot \frac{{k-1 \choose i-1} + {k-1 \choose i}}{a + \gamma \cdot i} \\
    &= \sum_{i = 0}^{k-1} (-1)^i \cdot \frac{{k-1 \choose i}}{a + \gamma \cdot i} + \sum_{i = 1}^{k} (-1)^i \cdot \frac{{k-1 \choose i-1}}{a + \gamma \cdot i} \\
    &= \sum_{i = 0}^{k-1} (-1)^i \cdot \frac{{k-1 \choose i}}{a + \gamma \cdot i} - \sum_{j = 0}^{k-1} (-1)^j \cdot \frac{{k-1 \choose j}}{(a+\gamma) + \gamma \cdot j},
\end{align}
    where we have set $j = i-1$. We can now use the inductive hypothesis on $k-1$ to obtain that this equals
\begin{align}
    & \frac{(k-1)! \cdot \gamma^{k-1}}{a(a+\gamma) \cdots (a+(k-1) \gamma)} - \frac{(k-1)! \cdot \gamma^{k-1}}{(a+\gamma) \cdots (a+(k-1) \gamma) (a + k \gamma)} \\
    &= (k-1)! \cdot \gamma^{k-1} \cdot \frac{(a + k \gamma) - a}{a(a+\gamma) \cdots (a+(k-1) \gamma)(a+k \gamma)} \\
    &= k! \cdot \gamma^k \cdot \frac{1}{a(a+\gamma) \cdots (a+(k-1) \gamma)(a+k \gamma)}.
\end{align}

Therefore, by setting $a = 1$ and replacing $\gamma$ with $\gamma' = \gamma/k$, we have that the difference in support size between $D_1$ and $D_2$ is
\[\frac{n_0}{2^{k-1}} \cdot \frac{k!}{k^k} \cdot \gamma^k \cdot \frac{1}{(1+\gamma/k)(1+2\gamma/k) \cdot (1+\gamma)}.\]
    Assuming that $\gamma \le 1/2$, we can apply Stirling's approximation to obtain that this difference is $n_0 \cdot (\gamma/\Theta(1))^k$.
    
    To finish, we will set $n_0$ appropriately. Note that we wish for $D_2$ to have support size exactly $n$. However, all of the points in $D_2$ has mass between $\frac{1}{n_0}$ and $\frac{1+\gamma}{n_0},$ which means that the support size $n_2$ must be between $\frac{n_0}{1+\gamma}$ and $n_0$. So, we can first set $n_0$, and then choose $n = n_2$ to be $n_0 \cdot \sum_{i = 0, i \text{ odd}}^k \frac{{k \choose i}}{2^{k-1} \cdot (1 + \frac{\gamma \cdot i}{k})},$ which is in the range $\left[\frac{n_0}{1+\gamma}, n_0\right]$, and $n_1$ to be $n_0 \cdot \sum_{i = 0, i \text{ even}}^k \frac{{k \choose i}}{2^{k-1} \cdot (1 + \frac{\gamma \cdot i}{k})}.$ Both $n_1$ and $n = n_2$ are in the range $\left[\frac{n_0}{1+\gamma}, n_0\right]$. Indeed, we will have that $\sum_{i = 1}^{n_1} (p_{D_1}(i))^\ell = \sum_{i = 1}^{n_2} (p_{D_2}(i))^\ell$, and $n_1-n_2 = \Theta(\gamma)^k \cdot n_0 = \Theta(\gamma)^k \cdot n$. In addition, because all of the values $p_{D_j}(i)$ are in the range $\left[\frac{1}{n_0}, \frac{1+\gamma}{n_0}\right]$ for both $j = 1$ and $j = 2$, this means that they are also in the range $\left[\frac{1-\gamma}{n}, \frac{1+\gamma}{n}\right]$, as desired. This completes the proof.
\end{proof}

\section{Proof of \texorpdfstring{\Cref{lem:expression-mean-zero}}{Technical Lemma}}
\label{sec:proof-technical}

The goal of this section is prove \Cref{lem:expression-mean-zero} which will finish the proof of \Cref{lem:estimate-sum}.
{
    \def\thetheorem{\ref{lem:expression-mean-zero}}
    \begin{lemma}
        For all sequences of numbers $(\beta_j)_{j \in [\ssize]}$ and all $\alpha$ the following identity holds:
        \begin{align*}
            \sum_{\substack{I \subseteq [\ssize]\\0 < |I| \le k}} (-1)^{|I| + 1} \frac{\binom{k}{|I|}}{\binom{\ssize}{|I|}} \left(\prod_{j \in I} \beta_j - (1 + \alpha)^h\right)
                = (-1)^{k + 1} \sum_{\substack{I \subseteq [\ssize]\\0 < |I| \le k}} \frac{\binom{k}{|I|}}{\binom{m}{|I|}} \alpha^{k - |I|} \prod_{j \in I} (\beta_j - (1 + \alpha))
        \end{align*}
    \end{lemma}
    \addtocounter{theorem}{-1}
}

First we need a simple lemma.
\begin{lemma}\label{lem:prod-mean-zero}
    For all sequences of numbers $(\beta_j)_{j \in I}$ and all $\alpha$ the following identity holds:
    \begin{align*}
        \prod_{j \in I} \beta_j - (1 + \alpha)^{|I|}
            = \sum_{\emptyset \neq J \subseteq I} (1 + \alpha)^{|I| - |J|} \prod_{j \in J} (\beta_j - (1 + \alpha))
    \end{align*}
\end{lemma}
\begin{proof}
    Let $\beta_j' := \beta_j-(1+\alpha)$. Then, 
    \begin{align*}
        \prod_{j \in I} \beta_j' - (1 + \alpha)^{|I|}
        &= \prod_{j \in I} (\beta_j' + (1+\alpha)) - (1 + \alpha)^{|I|} \\
        &=\left[\sum_{J \subseteq I} (1 + \alpha)^{|I| - |J|} \prod_{j \in J} \beta_j'\right] - (1+\alpha)^{|I|} \\
        &= \sum_{\emptyset \neq J \subseteq I} (1 + \alpha)^{|I| - |J|} \prod_{j \in J} (\beta_j - (1 + \alpha)). \qedhere
    \end{align*}
\end{proof}

Now we are ready to prove \Cref{lem:expression-mean-zero}.
\begin{proof}[Proof of \Cref{lem:expression-mean-zero}]
        
    We start by using \Cref{lem:prod-mean-zero}
    \begin{align*}
        \sum_{\substack{I \subseteq [\ssize]\\0 < |I| \le k}} (-1)^{|I| + 1} \frac{\binom{k}{|I|}}{\binom{\ssize}{|I|}} \left(\prod_{j \in I} \beta_j - (1 + \alpha)^h\right)
            &= \sum_{\substack{I \subseteq [\ssize]\\0 < |I| \le k}} (-1)^{|I| + 1} \frac{\binom{k}{|I|}}{\binom{\ssize}{|I|}} \sum_{\emptyset \neq J \subseteq I} (1 + \alpha)^{|I| - |J|} \prod_{j \in J} (\beta_j - (1 + \alpha))
    \end{align*}
    We use that $\frac{\binom{k}{|I|}}{\binom{\ssize}{|I|}}
    = \frac{\binom{\ssize - |I|}{k - |I|}}{\binom{\ssize}{k}}$ which follows from the fact that $\binom{\ssize}{k}\binom{k}{|I|} = \binom{\ssize}{|I|, k - |I|, \ssize - k} = \binom{\ssize}{|I|}\binom{\ssize - |I|}{k - |I|}$.
    \begin{align*}
        &\sum_{\substack{I \subseteq [\ssize]\\0 < |I| \le k}} (-1)^{|I| + 1} \frac{\binom{k}{|I|}}{\binom{\ssize}{|I|}} \sum_{\emptyset \neq J \subseteq I} (1 + \alpha)^{|I| - |J|} \prod_{j \in J} (\beta_j - (1 + \alpha))
            \\&= \sum_{\substack{I \subseteq [\ssize]\\0 < |I| \le k}} (-1)^{|I| + 1} \frac{\binom{\ssize - |I|}{k - |I|}}{\binom{\ssize}{k}} \sum_{\emptyset \neq J \subseteq I} (1 + \alpha)^{|I| - |J|} \prod_{j \in J} (\beta_j - (1 + \alpha))
            \\&= \frac{1}{\binom{\ssize}{k}} \sum_{\substack{J \subseteq [\ssize]\\0 < |J| \le k}} \left(\prod_{j \in J} (\beta_j - (1 + \alpha)) \right) \sum_{\substack{I \subseteq [\ssize]\\J \subseteq I, |I| \le k}} (-1)^{|I| + 1} \binom{\ssize - |I|}{k - |I|} (1 + \alpha)^{|I| - |J|}
            \\&= \frac{1}{\binom{\ssize}{k}} \sum_{\substack{J \subseteq [\ssize]\\0 < |J| \le k}} \left( \prod_{j \in J} (\beta_j - (1 + \alpha)) \right) \sum_{h = |J|}^{k} (-1)^{h + 1} \binom{\ssize - |J|}{h - |J|}\binom{\ssize - h}{k - h} (1 + \alpha)^{h - |J|}
    \end{align*}
    Now we use that $\binom{\ssize - h}{k - h} \binom{\ssize - |J|}{h - |J|} = \binom{\ssize - |J|}{k - h, \ssize - k, h - |J|} = \binom{\ssize - |J|}{k - |J|} \binom{k - |J|}{h - |J|}$.
    \begin{align*}
        &\frac{1}{\binom{\ssize}{k}} \sum_{\substack{J \subseteq [\ssize]\\0 < |J| \le k}} \left( \prod_{j \in J} (\beta_j - (1 + \alpha)) \right) \sum_{h = |J|}^{k} (-1)^{h + 1} \binom{\ssize - |J|}{h - |J|}\binom{\ssize - h}{k - h} (1 + \alpha)^{h - |J|} 
        \\&= \frac{1}{\binom{\ssize}{k}} \sum_{\substack{J \subseteq [\ssize]\\0 < |J| \le k}} \left( \prod_{j \in J} (\beta_j - (1 + \alpha)) \right) \sum_{h = |J|}^{k} (-1)^{h + 1} \binom{\ssize - |J|}{k - |J|} \binom{k - |J|}{h - |J|} (1 + \alpha)^{h - |J|} 
        \\&= \frac{1}{\binom{\ssize}{k}} \sum_{\substack{J \subseteq [\ssize]\\0 < |J| \le k}} (-1)^{|J| + 1} \binom{\ssize - |J|}{k - |J|} \left( \prod_{j \in J} (\beta_j - (1 + \alpha)) \right) \sum_{h = 0}^{k - |J|} (-1)^{h}  \binom{k - |J|}{h} (1 + \alpha)^{h}
        \\&= \frac{1}{\binom{\ssize}{k}} \sum_{\substack{J \subseteq [\ssize]\\0 < |J| \le k}} (-1)^{|J| + 1} \binom{\ssize - |J|}{k - |J|} \left( \prod_{j \in J} (\beta_j - (1 + \alpha)) \right) (1 - (1 + \alpha))^{k - |J|}
        \\&= \frac{(-1)^{k + 1}}{\binom{\ssize}{k}} \sum_{\substack{J \subseteq [\ssize]\\0 < |J| \le k}} \binom{\ssize - |J|}{k - |J|} \alpha^{k - |J|} \left( \prod_{j \in J} (\beta_j - (1 + \alpha)) \right)
    \end{align*}
    Finally, we again use that $\frac{\binom{k}{|I|}}{\binom{\ssize}{|I|}}
    = \frac{\binom{\ssize - |I|}{k - |I|}}{\binom{\ssize}{k}}$ to finish the proof.
    \begin{align*}
        \frac{(-1)^{k + 1}}{\binom{\ssize}{k}} \sum_{\substack{J \subseteq [\ssize]\\0 < |J| \le k}} \binom{\ssize - |J|}{k - |J|} \alpha^{k - |J|} \left( \prod_{j \in J} (\beta_j - (1 + \alpha)) \right)
            = (-1)^{k + 1} \sum_{\substack{J \subseteq [\ssize]\\0 < |J| \le k}} \frac{\binom{k}{h}}{\binom{\ssize}{h}} \alpha^{k - |J|} \left( \prod_{j \in J} (\beta_j - (1 + \alpha)) \right)
    \end{align*}
\end{proof}

\bibliographystyle{plainnat}
\bibliography{literature}
\end{document}